\documentclass[sigconf]{acmart}

\setcopyright{acmcopyright}
\copyrightyear{2022}
\acmYear{2022}
\acmDOI{}

\acmConference[]{}{}{}
\acmBooktitle{}
\acmPrice{}
\acmISBN{}

\usepackage{amsthm,amsmath}
\usepackage{tikz}

\def\x#1{\textbf{#1}}

\usepackage{mathtools}
\usepackage{enumerate}

\newtheorem{theorem}{Theorem}

\newtheorem{thm}[theorem]{Theorem}

\newtheorem{lemma}[theorem]{Lemma}

\newtheorem{defi}[theorem]{Definition}
\newtheorem{example}[theorem]{Example}

\newtheorem{convention}[theorem]{Convention}

\def\syz{\operatorname{Syz}}

\renewcommand{\iff}{\Leftrightarrow}

\makeatletter
\newcommand{\myitem}[1]{%
\item[(#1)]\protected@edef\@currentlabel{#1}%
}
\makeatother

\def\eatspace#1{#1}
\def\step#1#2{\par\kern1pt\hangindent#2em\hangafter=1\noindent\rlap{\small#1}\kern#2em\relax\eatspace}

\let\set\mathbb
\def\<#1>{\langle#1\rangle}

\def\id{\operatorname{id}}

\def\lc{\operatorname{lc}}

\copyrightyear{2022}
\acmYear{2022}
\setcopyright{acmcopyright}\acmConference[ISSAC '22]{Proceedings of the 2022 International Symposium on Symbolic and Algebraic Computation}{July 4--7, 2022}{Villeneuve-d'Ascq, France}
\acmBooktitle{Proceedings of the 2022 International Symposium on Symbolic and Algebraic Computation (ISSAC '22), July 4--7, 2022, Villeneuve-d'Ascq, France}
\acmPrice{15.00}
\acmDOI{10.1145/3476446.3536187}
\acmISBN{978-1-4503-8688-3/22/07}

\begin{document}
\fancyhead{}
\title{Order-Degree-Height Surfaces for Linear Operators}
\titlenote{H. Huang was supported by the NSFC grant (No.\ 12101105) and 
    the Fundamental Research Funds for the Central Universities.
    M. Kauers and G. Mukherjee were supported by the Austrian FWF grants P31571-N32 and W1214, project part~13.
}

\author[H. Huang]{Hui Huang}
\affiliation{%
  \institution{School of Mathematical Sciences, \\Dalian University of Technology}
  \city{Dalian, Liaoning, 116024}
  \state{}
  \postcode{}
  \country{China}
}
\email{huanghui@dlut.edu.cn}

\author[M. Kauers]{Manuel Kauers}
\author[G. Mukherjee]{Gargi Mukherjee}
\affiliation{%
  \institution{Institute for Algebra, Johannes Kepler University}
  \city{Linz, A4040}
  \state{}
  \postcode{}
  \country{Austria}
}
\email{manuel.kauers@jku.at}
\email{gargi.mukherjee@jku.at}

\begin{abstract}
  It is known for linear operators with polynomial coefficients annihilating a given D-finite function that
  there is a trade-off between order and degree.
  Raising the order may give room for lowering the degree.
  The relationship between order and degree is typically described by a hyperbola
  known as the order-degree curve.
  In this paper, we add the height into the picture, i.e., a measure for the size
  of the coefficients in the polynomial coefficients.
  For certain situations, we derive relationships between order, degree, and height
  that can be viewed as order-degree-height surfaces.
\end{abstract}


\maketitle

\section{Introduction}


D-finiteness and operations preserving D-finiteness play an important role
in computer algebra. A function is called D-finite if it is annihilated by
a nonzero linear operator with polynomial coefficients. Annihilating operators
are used to represent D-finite functions, and to perform operations with them. It
is therefore of interest to know how big such operators are. The primary
interest is in the order of the operator. An operator of minimal order is
a generator of the ideal of all annihilating operators for a given D-finite function, and certain
facts about the function, for example its asymptotic behaviour, can be
most easily extracted from such an operator. A disadvantage of a minimal
order operator is that it may not be the smallest operator in terms of arithmetic size. The secondary
interest is in the degrees of the polynomial coefficients of the operator.
The arithmetic size is roughly the product of order and degree, and the
minimal arithmetic size is often not achieved by the operator of minimal
order. Instead, there is the possibility of trading order against degree:
allowing a slightly higher order of the
operator can lead to a substantial degree drop which altogether results in
an operator of smaller arithmetic size. This effect has been observed and
analyzed for many different operations related to linear operators
\cite{BCLSS2007,ChKa2012a,ChKa2012b,CJKS2013,Kaue2014}.
Typical bounds are formulated in terms of hyperbolas such that for every
point $(r,d)$ above the hyperbola, the function at hand has an
annihilating operator of order $r$ and degree~$d$. These hyperbolas are
known as order-degree curves.

Besides order and degree, the size of an operator depends on a third parameter,
called the height. It measures the size of the coefficients of the polynomial
coefficients of the operator. For example, for differential operators
over polynomials over the integers, it matters how long the integers 
appearing in the operator are. If we define the height as the log of the largest integer
(in absolute value), the bitsize of an operator is roughly the product of its
order, its degree, and its height. Besides some bounds for 
hypergeometric creative telescoping~\cite{KaYe2015} and D-finite closure properties~\cite{Kaue2014},
not much is known about the bit size of operators that arise from computations
with D-finite functions. The
experience is however that the point $(r,d)$ on the order-degree curve which
minimizes the arithmetic size does not also minimize the bitsize. But then,
which point on the curve does minimize the bitsize? Does the point of minimal
bitsize even sit on the curve, or can it happen that a further reduction of
the bitsize is possible by increasing both the order and the degree? These were
the motivating questions for the present paper.

Our goal is to extend the theory of order-degree curves to order-degree-height
surfaces. Instead of hyperbolas that describe the points $(r,d)$ such that a
given function has annihilating operators of order~$r$ and degree~$d$, we will
see surfaces defined by polynomials describing the points $(r,d,h)$
such that the function has annihilating operators of order~$r$, degree~$d$, and height~$h$.
Unfortunately, we do not have any such results for operators in $\set Z[x][\partial]$, with
the height defined by the longest integer. In order to be able to apply the
techniques from linear algebra that were used for deriving order-degree curves,
we instead consider operators in $C[y][x][\partial]$ where $C$ is a field, with
the height defined by the degree in~$y$. Our results in this setting suggest
that order and degree can indeed be traded against height, and experiments with
actual operators confirm this qualitative prediction, even though our bounds
quantitatively are not very tight. Some of our bounds are particularly bad for
large orders~$r$. While for a fixed degree and $r\to\infty$, we would usually
expect a bound on the height that converges to a constant, two of our bounds
grow quadratically in~$r$. This may be an indication that the linear algebra
approach which was used for deriving tight order-degree curves is not sufficient
for deriving sharp order-degree-height surfaces. 

The remainder of the paper proceeds as follows. In Section~\ref{SEC:prelim}, 
we provide background and basic notions required in the paper. We then derive
order-degree-height surfaces for three different situations:
(a)~in Section~\ref{SEC:clm} for common left multiples of some given operators, extending the work of~\cite{Kaue2014},
(b)~in Section~\ref{SEC:ct} for hypergeometric creative telescoping, extending the work of~\cite{ChKa2012a}, and
(c)~in Section~\ref{SEC:desing} for contraction ideals of operators, extending the work of~\cite{CJKS2013}.
In all these cases, we derive polynomials $p$ such that for all $r,d,h$ with $p(r,d,h)\geq0$ there
exists an operator of order~$r$, degree~$d$, and height~$h$ with the desired feature.
The boundary of the region defined by this polynomial inequality is what we call the order-degree-height
surface.

\section{Notation and General Assumptions}\label{SEC:prelim}

The notation $R[\partial]$ already used in the introduction refers to an Ore
algebra over the commutative ring~$R$. Recall that such an Ore algebra is
defined in terms of an endomorphism $\sigma\colon R\to R$ and a $\sigma$-derivation
$\delta\colon R\to R$, i.e., a map satisfying $\delta(a+b)=\delta(a)+\delta(b)$
and $\delta(ab)=\delta(a)b+\sigma(a)\delta(b)$ for all $a,b\in R$.
The addition in $R[\partial]$ is defined as for polynomials, and the
multiplication in $R[\partial]$ extends the multiplication on $R$ through
the commutation rule $\partial a=\sigma(a)\partial+\delta(a)$ ($a\in R$).
An element $a$ of $R$ is called a constant if $\sigma(a)=a$ and $\delta(a)=0$.

Canonical examples for Ore algebras include differential operators, where
$\sigma=\id$ and $\delta$ is a derivation on~$R$, and difference operators,
where $\sigma$ is an automorphism on $R$ and $\delta=0$. Throughout this
paper, $R$ will be a ring of polynomials or rational functions in two variables
$x,y$ over a field~$C$ of characteristic zero, e.g. $R=C[x,y]$ or $R=C(x,y)$ or $R=C(x)[y]$, etc.
We assume throughout that the elements of $C[y]$ (and thus $C(y)$) are constants. We further
assume that $\sigma$ is an automorphism of $R$, 
$\sigma$ and $\delta$ map polynomials to polynomials, and that
they do not increase degrees. More precisely, we assume
\begin{alignat*}1
  &\deg_x(\sigma(f))=\deg_x(f),\quad\deg_x(\delta(f))\leq\deg_x(f),\\
  &\deg_y(\sigma(f))=\deg_y(f),\quad\deg_y(\delta(f))\leq\deg_y(f)
\end{alignat*}
for all $f\in C[x,y]$.
This is true for example for differential operators with the derivation $\delta=\frac{d}{dx}$
and for difference operators with the shift $\sigma(x)=x+1$.

The order of an element $L\in R[\partial]$ is defined as $\deg_\partial(L)$.
The height of an element $p\in C[x,y]$ is defined as $\deg_y(p)$, and
$\deg_x(p)$ is called its degree. The height/degree of a rational
function $p/q$ is defined as $\deg(p)-\deg(q)$ (for $\deg=\deg_y$ and $\deg=\deg_x$,
respectively). If $\sigma=\id$ and $\delta=\frac{d}{dx}$, we write $D_x$ instead of~$\partial$,
and if $\sigma(x)=x+1$ and $\delta=0$, we write $S_x$ instead of~$\partial$.

Elements of an Ore algebra $R[\partial]$ can be interpreted as operators.
A function space $F$ for $R[\partial]$ is defined as a left $R[\partial]$-module.
For a fixed $f\in F$, the set of all $L\in R[\partial]$ with $L\cdot f=0$ forms
a left ideal of $R[\partial]$, called the annihilator of~$f$. As we will only
consider left ideals in this paper, we will drop the attribute `left' from
now on. If $R$ is a field, then $R[\partial]$ is a left-Euclidean domain, so
every ideal is generated by a single element. In this case, any nonzero
ideal element of minimal order can serve as a generator. We call such elements
minimal. They are pairwise associate, i.e., any two minimal elements of an
ideal are left-$R$-multiples of each other.

\section{Common Left Multiples}\label{SEC:clm}

A common left multiple of some operators $L_1,\dots,L_n\in R[\partial]$ is
an operator $L\in R[\partial]$ such that there exist operators $M_1,\dots,M_n\in R[\partial]$
with $L = M_1L_1=\cdots=M_nL_n$.
If $L_1,\dots,L_n$ are annihilating operators of certain functions $f_1,\dots,f_n$,
then $L$ is an annihilator for every $C$-linear combination $\alpha_1f_1+\cdots+\alpha_nf_n$
of these functions. 
Bounds on orders and degrees of common left multiples were given in~\cite{BCLS2012}. 
An order-degree curve for this case has appeared in~\cite{Kaue2014}. 
Continuing the discussion in Section~2.2 of~\cite{Kaue2014}, %
we provide an order-degree-height surface describing the shapes of common
left multiples of $L_1,\dots,L_n$.

\begin{thm}\label{thm:lclm}
  Let $L_1,\dots,L_n\in C[x,y][\partial]$ and let
  $r_\ell=\deg_{\partial}(L_\ell)$,
  $d_\ell=\deg_x(L_\ell)$,
  $h_\ell=\deg_y(L_\ell)$
  for $\ell=1,\dots,n$.
  Let $r,d,h\in \set N$ be such that
  \begin{alignat*}1
    &(r+1)(d+1)(h+1)\\
    &- (r+1)(d+1)\sum_{\ell=1}^n h_\ell \ + \ (h+1)\sum_{\ell=1}^n r_\ell d_\ell\\
    &- (r+1)(h+1)\sum_{\ell=1}^n d_\ell \ + \ (d+1)\sum_{\ell=1}^n r_\ell h_\ell\\
    &- (d+1)(h+1)\sum_{\ell=1}^n r_\ell \ + \ (r+1)\sum_{\ell=1}^n d_\ell h_\ell
     - \sum_{\ell=1}^n r_\ell d_\ell h_\ell > 0.
  \end{alignat*}
  Then there exists a common left multiple of $L_1,\dots,L_n$ in $C[x,y][\partial]$
  of order at most~$r$, degree at most~$d$, and height at most~$h$.
\end{thm}
\begin{proof}
  Make an ansatz for $n$ operators
  \[
    M_\ell = \sum_{i=0}^{r-r_\ell}\sum_{j=0}^{d-d_\ell}\sum_{k=0}^{h-h_\ell} m_{i,j,k,\ell} y^kx^j\partial^i
  \]
  with undetermined coefficients $m_{i,j,k,\ell}$ and equate the coefficients of $M_\ell L_\ell$
  for $\ell=1,\dots,n$ to each other in order to obtain a common left multiple of the desired shape.
  This leads to 
  \begin{alignat*}1
    &\sum_{\ell=1}^n\sum_{i=0}^{r-r_\ell}\sum_{j=0}^{d-d_\ell}\sum_{k=0}^{h-h_\ell}1
  \end{alignat*}
  variables and $(n-1)(r+1)(d+1)(h+1)$ equations.
  If $r,d,h$ are as in the statement of the theorem, the linear system has more variables than
  equations and therefore a nonzero solution. 
\end{proof}

In the proof of the theorem, not only the input and output operators are restricted to $C[x,y][\partial]$,
but also the multipliers $M_1,\dots,M_n$. In general, allowing the $M_1,\dots,M_n$ to be in $C(x,y)[\partial]$
can lead to common multiples of lower degree or height. 

\begin{example}
  We consider two randomly chosen operators $L_1,L_2\in\set Q[x,y][D_x]$ with
  $\deg_{D_x}(L_1)=2$, $\deg_{D_x}(L_2)=1$, $\deg_{x}(L_1)=1$,
  $\deg_{x}(L_2)=2$, $\deg_{y}(L_1)=1$, $\deg_{y}(L_2)=1$.
  In the following table, we compare the actual sizes of common left multiples of $L_1,L_2$
  with the sizes predicted by Thm.~\ref{thm:lclm}.
  A table entry $u|v$ in column~$r$ and row~$d$ 
  means that Thm.~\ref{thm:lclm} predicts a common multiple of order~$r$, degree~$d$ and any height $h\geq u$,
  and for the specific operators $L_1,L_2$ we found a common multiple of order~$r$, degree~$d$, and height $h=v$.
  The common multiples were computed in $C(x,y)[D_x]$. 
  A dot means that no operator of the respective order and degree was found or predicted.
  In this experiment, Thm.~\ref{thm:lclm} predicts orders and degrees correctly but the bound on the height overshoots.
  At the same time, it does at least qualitatively reflect the effect that increasing both $r$ and $d$
  makes it possible to decrease~$h$.
  We have repeated the experiment with operators $L_1,L_2$ of some other shapes and always found a similar conclusion.
  Note that in this particular example, Thm.~\ref{thm:lclm} also rightly predicts the symmetry w.r.t.\ $r$ and~$d$. 
  \[\scriptscriptstyle
  \begin{array}{@{}c|c@{\kern5pt}c@{\kern5pt}c@{\kern5pt}c@{\kern5pt}c@{\kern5pt}c@{\kern5pt}c@{\kern5pt}c@{\kern5pt}c@{\kern5pt}c@{\kern5pt}c@{\kern5pt}c@{}}
    d \\
    10& {\cdot|\cdot} & {\cdot|\cdot} & {\cdot|\cdot} & {15|5} & {6|4} & {4|3} & {3|3} & {3|3} & {3|3} & {3|2} & {3|2} \\
    9& {\cdot|\cdot} & {\cdot|\cdot} & {\cdot|\cdot} & {21|5} & {6|5} & {4|4} & {4|3} & {3|3} & {3|3} & {3|3} & {3|2} \\
    8& {\cdot|\cdot} & {\cdot|\cdot} & {\cdot|\cdot} & {37|5} & {7|5} & {5|4} & {4|3} & {3|3} & {3|3} & {3|3} & {3|3} \\
    7& {\cdot|\cdot} & {\cdot|\cdot} & {\cdot|\cdot} & {\cdot|\cdot} & {9|6} & {5|4} & {4|3} & {4|3} &  {3|3} & {3|3} & {3|3} \\
    6& {\cdot|\cdot} & {\cdot|\cdot} & {\cdot|\cdot} & {\cdot|\cdot} & {12|8} & {7|5} & {5|4} & {4|3} & {4|3} & {4|3} & {3|3} \\
    5& {\cdot|\cdot} & {\cdot|\cdot} & {\cdot|\cdot} & {\cdot|\cdot} & {31|19} & {10|7} & {7|5} & {5|4} & {5|4} & {4|4} & {4|3} \\
    4& {\cdot|\cdot} & {\cdot|\cdot} & {\cdot|\cdot} & {\cdot|\cdot} & {\cdot|\cdot} & {31|19} & {12|8} & {9|6} & {7|5} & {6|5} & {6|4} \\
    3& {\cdot|\cdot} & {\cdot|\cdot} & {\cdot|\cdot} & {\cdot|\cdot} & {\cdot|\cdot} & {\cdot|\cdot} & {\cdot|\cdot} & {\cdot|\cdot} & {37|5} & {21|5} & {15|5} \\
    2& {\cdot|\cdot} & {\cdot|\cdot} & {\cdot|\cdot} & {\cdot|\cdot} & {\cdot|\cdot} & {\cdot|\cdot} & {\cdot|\cdot} & {\cdot|\cdot} & {\cdot|\cdot} & {\cdot|\cdot} & {\cdot|\cdot} \\
    1& {\cdot|\cdot} & {\cdot|\cdot} & {\cdot|\cdot} & {\cdot|\cdot} & {\cdot|\cdot} & {\cdot|\cdot} & {\cdot|\cdot} & {\cdot|\cdot} & {\cdot|\cdot} & {\cdot|\cdot} & {\cdot|\cdot} \\
    0& {\cdot|\cdot} & {\cdot|\cdot} & {\cdot|\cdot} & {\cdot|\cdot} & {\cdot|\cdot} & {\cdot|\cdot} & {\cdot|\cdot} & {\cdot|\cdot} & {\cdot|\cdot} & {\cdot|\cdot} & {\cdot|\cdot} \\\hline
    & 0 & 1 & 2 & 3 & 4 & 5 & 6 & 7 & 8 & 9 & 10 & r 
  \end{array}
  \]
\end{example}

\section{Creative telescoping}\label{SEC:ct}

In this section, we consider bivariate proper hypergeometric terms
and aim at describing the relationships between orders, degrees and heights
for their telescopers. Bounds on these three quantities have been given in \cite{MoZe2005,ChKa2012a,KaYe2015}.
An order-degree curve for this case was also provided in \cite{ChKa2012a}.
We will closely follow the discussion in~\cite{ChKa2012a} and further extend the curve
to an order-degree-height surface.

In this respect, we need to generalize all related notions introduced
previously to the bivariate case. Let $k$ be a new variable distinct
from the two variables $x,y$ and let $C(x,y,k)$ be the field of
rational functions in $x,y,k$. Let $S_x$ and $S_k$ denote the
difference operators on $C(x,y,k)$ defined by
\[
S_x(f(x,y,k)) = f(x+1,y,k)
\ \ \text{and}\ \
S_k(f(x,y,k)) = f(x,y,k+1)
\]
for any rational function $f\in C(x,y,k)$. Clearly, $S_x$ and $S_k$
commute with each other and leave elements in $C(y)$ fixed.

We will be considering some extension field $E$ of $C(x,y,k)$ on which
difference operators $S_x$ and $S_k$ can be naturally extended. A
{\em hypergeometric term} with respect to $x,k$ is a nonzero element $H$ of
such an extension field $E$ with $S_x(H)/H \in C(x,y,k)$ and
$S_k(H)/H\in C(x,y,k)$. 
We are particularly interested in proper hypergeometric terms, that is,
hypergeometric terms which can be rewritten in the following form
\begin{equation}\label{EQ:hyper}
H = p \alpha^x \beta^k \prod_{m=1}^M
\frac{\Gamma(a_mx+a_m'k+a_m'')\Gamma(b_mx-b_m'k+b_m'')}
{\Gamma(u_mx+u_m'k+u_m'')\Gamma(v_mx-v_m'k+v_m'')},
\end{equation}
where $p\in C[x,y,k]\setminus\{0\}$, $\alpha,\beta\in C[y]\setminus\{0\}$, $M\in \set N$,
$a_m,a_m'$, $b_m,b_m'$, $u_m,u_m'$, $v_m,v_m'\in \set N$,
$a_m'',b_m'',u_m'',v_m''\in C[y]$. 

Let $H$ be a proper hypergeometric term of the form \eqref{EQ:hyper}.
The method of creative telescoping consists in finding polynomials
$c_0,c_1,\dots,c_r\in C[x,y]$, not all zero, and a rational function
$g\in C(x,y,k)$ such that
\[
c_0 H + c_1 S_x(H)+\cdots+c_rS_x^r(H) = S_k(gH) - gH.
\]
If $c_0,c_1,\dots,c_r$ and $g$ are as above, then we say that the
operator $L = c_0+c_1S_x+\cdots+c_rS_x^r\in C[x,y][S_x]$ is a
{\em telescoper} for $H$ and $g$ is a {\em certificate} for $L$ (and
$H$). According to the fundamental lemma in \cite{WiZe1992a}, the term
$H$ always admits a telescoper. All telescopers for $H$ form an
ideal in $C[x,y][S_x]$.

The main task of this section is to derive a surface in the
$(r,d,h)$-space which indicates that for every integer point $(r,d,h)$
above the surface, there is a telescoper for $H$ of order $r$, degree $d$ and height~$h$.
Following \cite{ChKa2012a}, we make a case
distinction between non-rational input and rational input.

\subsection{The non-rational case}
In this subsection, we assume that the following holds.
\begin{convention}\label{CON:nonrat}
Let $H$ be a proper hypergeometric term of the form \eqref{EQ:hyper},
and assume that $H$ cannot be split into $H = fH_0$ for $f\in
C(x,y,k)$ and another proper hypergeometric term $H_0$ with
$S_k(H_0)/H_0 = 1$.
\end{convention}
Informally speaking, the above assumption excludes those terms of the
form \eqref{EQ:hyper} in which $\beta = 1$ and every $\Gamma$-term
involving $k$ can be cancelled against another one up to some rational
function. Those exceptional terms are treated separately in
Section~\ref{SUBSEC:rat} below.

With Convention~\ref{CON:nonrat}, for any rational function $g\in
C(x,y,k)\setminus\{0\}$, we know that $gH$ cannot be split into the indicated way
either. In particular, $S_k(gH)/(gH)\neq 1$, that is, $S_k(gH) -
gH \neq 0$. As a consequence, whenever the pair $(L,g)$ with $L\in C[x,y][S_x]$
and $g\in C(x,y,k)\setminus\{0\}$ is such that $L(H) = S_k(gH) - gH$, we can
guarantee that $L$ is not the zero operator. Thus, in the sequel of
this subsection, we need not worry about this requirement any more.

The main idea we are using for analysis is the same
as \cite{MoZe2005,ChKa2012a}, that is, to go through steps of
Zeilberger's algorithm \cite{Zeil1991} literally when applied to the
given proper hypergeometric term $H$, and then reduce the problem of computing a telescoper
to solving a homogeneous system of linear equations over~$C$, which will have
a nontrivial solution whenever it is underdetermined.

With Convention~\ref{CON:nonrat}, for any integer $m$ with $1\leq
m\leq M$, define
\begin{align*}
&A_m = a_mx+a_m'k+a_m'',\quad
B_m = b_mx-b_m'k+b_m'',\\
&U_m = u_mx+u_m'k+u_m'',\quad
V_m = v_mx-v_m'k+v_m''.
\end{align*}
For $p\in C[x,y,k]$ and $m\in \set N$, let $p^{\overline{m}} =
p(p+1)\dots(p+m-1)$ with the conventions $p^{\overline{0}}=1$ and
$p^{\overline{1}}=p$.
For an operator $L = c_0 + c_1 S_x + \cdots + c_r S_x^r\in
C[x,y][S_x]$, we then have
\begin{align*}
&L(H) = \sum_{i=0}^rc_i \alpha^i\frac{S_x^i(p)}p\prod_{m=1}^M
\frac{A_m^{\overline{ia_m}}B_m^{\overline{ib_m}}}
{U_m^{\overline{iu_m}}V_m^{\overline{iv_m}}}H
\\[1ex]
&= \left(\sum_{i=0}^rc_i\alpha^iS_x^i(p)\prod_{m=1}^MP_{i,m}\right)
\alpha^x\beta^k\prod_{m=1}^M
\frac{\Gamma(A_m)\Gamma(B_m)}{\Gamma(U_m+ru_m)\Gamma(V_m+rv_m)},
\end{align*}
where 
\[
P_{i,m}=A_m^{\overline{ia_m}}B_m^{\overline{ib_m}}
(U_m+iu_m)^{\overline{(r-i)u_m}}(V_m+iv_m)^{\overline{(r-i)v_m}}.
\]
We can write
\begin{equation}\label{EQ:gosperform}
\frac{S_k(L(H))}{L(H)} = \frac{S_k(P)}P\frac{Q}{S_k(R)},
\end{equation}
where
\begin{align}
P &= \sum_{i=0}^rc_i\alpha^iS_x^i(p)\prod_{m=1}^MP_{i,m},
\label{EQ:P}\\[1ex]
Q &= \beta\prod_{m=1}^MA_m^{\overline{a_m'}}(V_m+rv_m-v_m')^{\overline{v_m'}},
\label{EQ:Q}\\[1ex]
R &= \prod_{m=1}^M(U_m+ru_m-u_m')^{\overline{u_m'}}B_m^{\overline{b_m'}}.
\label{EQ:R}
\end{align}
A remarkable feature of this decomposition is that the undetermined
coefficients $c_0,\dots,c_r$ only appear linearly in the polynomial
$P$ and yet do not appear at all in the polynomials $Q,R$.
Depending on the actual values of the parameters appearing in $H$,
the decomposition \eqref{EQ:gosperform} may or may not satisfy the
condition $\gcd(Q,S_k^i(R))=1$ for all $i\in \set N$, and thus the
corresponding equation
\begin{equation}\label{EQ:gosper}
P = QS_k(Y) - R Y
\end{equation}
is not necessarily the true Gosper equation. 
However, even in this case, it only means that some solutions may be overlooked,
and any nontrivial solution $Y\in C(x,y)[k]$ to this equation will still give rise to a 
correct telescoper for $H$ and a certificate. Since our main interest lies in bounding
the size of telescopers, it is sufficient to investigate under which circumstances the equation
\eqref{EQ:gosper} with the above choices of $P,Q,R$ has a solution.

We proceed by performing a coefficient comparison between both sides
of \eqref{EQ:gosper} with respect to powers of $x,y$ and $k$, rather
than merely with respect to powers of $k$ as in~\cite{MoZe2005} or powers
of $x,k$ as in~\cite{ChKa2012a}, yielding a linear system over $C$.  In
order to better express the number of variables and equations in this
system, we make the following notational convention.
\begin{convention}\label{CON:notation}
With Convention~\ref{CON:nonrat}, let 
\begin{align*}
&\vartheta_x = \deg_x(p),\quad \vartheta_y = \deg_y(p),\quad \vartheta_k = \deg_k(p),\\
&\mu = \max\Big\{\sum_{m=1}^M(a_m+b_m),\sum_{m=1}^M(u_m+v_m)\Big\},\\
&\nu = \max\Big\{\sum_{m=1}^M(a_m'+v_m'),\sum_{m=1}^M(u_m'+b_m')\Big\},\\
&\xi = \max\Big\{\deg_y(\alpha)+\sum_{m=1}^M(a_m\deg_y(A_m)+b_m\deg_y(B_m)),\\
&\hphantom{=\max\Big(\Big)}\sum_{m=1}^M(u_m\deg_y(U_m)+v_m\deg_y(V_m))\Big\},\\
&\eta = \max\Big\{\deg_y(\beta)+\sum_{m=1}^M(a_m'\deg_y(A_m)+v_m'\deg_y(V_m)),\\
&\hphantom{=\max\Big(\Big)}\sum_{m=1}^M(u_m'\deg_y(U_m)+b_m'\deg_y(B_m))\Big\}.
\end{align*}
\end{convention}
Note that these parameters are all nonnegative integers which only
depend on $H$ but not on $r, d$ or $h$.
\begin{lemma}\label{LEM:Pdeg}
With Convention~\ref{CON:notation}, let $d$ and $h$ be the degree and
height of $L$, respectively. Then
\begin{align*}
&\deg_x(P) \leq d+\vartheta_x+r\mu,
\quad
\deg_y(P) \leq h+\vartheta_y+r\xi,\\
\text{and}\quad &\deg_k(P) \leq \vartheta_k+r\mu.
\end{align*}
As a consequence, $P$ contains at most 
\begin{equation}\label{EQ:Pterm}
(d+\vartheta_x+r\mu+1)(h+\vartheta_y+r\xi+1)(\vartheta_k+r\mu+1)
\end{equation}
nonzero monomials in terms of $x,y,k$.
\end{lemma}
\begin{proof}
Observe that
\begin{align*}
&\deg_x(c_i) \leq d,\quad \deg_y(c_i) \leq h,\quad \deg_k(c_i) = 0,\\
&\deg_x(P_{i,m}) \leq ia_m+ib_m+(r-i)u_m+(r-i)v_m,\\
&\deg_y(P_{i,m}) \leq ia_m\deg_y(A_m)+ib_m\deg_y(B_m)\\
&\qquad\qquad\quad+(r-i)u_m\deg_y(U_m)+(r-i)v_m\deg_y(V_m),\\
&\deg_k(P_{i,m}) \leq ia_m+ib_m+(r-i)u_m+(r-i)v_m
\end{align*}
for all $i = 0,\dots, r$ and $m = 1,\dots, M$. The lemma follows
from \eqref{EQ:P}.
\end{proof}
As suggested in \cite{ChKa2012a}, the degrees of $Y$ in $x,y$ and $k$ will be chosen in such a
way that $QS_k(Y)-RY$ only contains monomials which are expected to
occur in $P$, so that no additional equations will appear.
\begin{lemma}\label{LEM:Ydeg}
With Convention~\ref{CON:notation}, assume that $Y \in C[x,y,k]$
satisfies
\begin{align*}
&\deg_x(Y) \leq \deg_x(P)-\nu,\quad \deg_y(Y) \leq \deg_y(P) - \eta,\\
\text{and}\quad &\deg_k(Y) \leq \deg_k(P) - \nu.
\end{align*}
Then the number of nonzero monomials in terms of $x,y,k$ appearing in 
$P - (QS_k(Y)-RY)$ is bounded by the integer defined in~\eqref{EQ:Pterm}.
\end{lemma}
\begin{proof}
It follows from \eqref{EQ:Q} and \eqref{EQ:R} that
\begin{align*}
&\max\{\deg_x(Q),\deg_x(R)\} =
\max\{\deg_k(Q),\deg_k(R)\} = \nu,\\
\text{and}\quad
&\max\{\deg_y(Q),\deg_y(R)\} = \eta.
\end{align*}
Thus the degrees of $P- (QS_k(Y)-RY)$ in $x,y,k$ are bounded above by
those of $P$, yielding the assertion by Lemma~\ref{LEM:Pdeg}.
\end{proof}
Now we are ready to present the main result of this subsection, 
which can be viewed as a natural generalization of the order-degree
curve from \cite[Theorem~5]{ChKa2012a}.
\begin{theorem}\label{THM:hyperodh}
With Convention~\ref{CON:notation}, let $r,d,h\in \set N$ be such
that
$d+\vartheta_x+r\mu-\nu\geq 0$, 
$h+\vartheta_y+r\xi-\eta\geq 0$, $\vartheta_k+r\mu-\nu\geq 0$, and
\begin{align*}
&(r+1)(d+1)(h+1)-(d+1+\vartheta_x+r\mu)(\vartheta_k+r\mu+1)\eta\\
&-(d+2+\vartheta_x+\vartheta_k+2r\mu-\nu)(h+1+\vartheta_y+r\xi-\eta)\nu
>0.
\end{align*}
Then there exists a telescoper for $H$ of order at most~$r$, degree at most~$d$ and
height at most~$h$.
\end{theorem}
\begin{proof}
In order to prove the theorem, it is sufficient to show that, for the
given triple $(r,d,h)$ in the assumption, the corresponding equation
$P = QS_k(Y) - RY$ has a nontrivial solution $Y \in C[x,y,k]$.
Lemma~\ref{LEM:Ydeg}, along with Lemma~\ref{LEM:Pdeg}, suggests the ansatz
\[
Y = \sum_{i=0}^{d+\vartheta_x+r\mu-\nu}
\,\sum_{j=0}^{h+\vartheta_y+r\xi-\eta}
\,\sum_{\ell=0}^{\vartheta_k+r\mu-\nu}
y_{i,j,\ell}x^iy^jk^\ell
\]
with undetermined coefficients $y_{i,j,\ell}$. 
Together with the unknown coefficients over $C$ of the $c_i$ in $P$ given by \eqref{EQ:P},
we obtain from the equation \eqref{EQ:gosper} a linear system over $C$ with
\begin{align*}
&(r+1)(d+1)(h+1)\\
&+(d+\vartheta_x+r\mu-\nu+1)(h+\vartheta_y+r\xi-\eta+1)(\vartheta_k+r\mu-\nu+1)
\end{align*}
variables and, by Lemma~\ref{LEM:Ydeg}, at most
\[
(d+\vartheta_x+r\mu+1)(h+\vartheta_y+r\xi+1)(\vartheta_k+r\mu+1)
\]
equations. By adding and removing $\eta$ in the second parenthesis of the number of equations,
along with a direct calculation, one readily sees that the linear system for $r,d,h$ satisfying the 
given constraints has more
variables than equations, ensuring the existence of such a nontrivial
solution. The assertion follows.
\end{proof}
\begin{example}\label{EX:nonrat}
  For the hypergeometric term $H=k\Gamma(x+k+y^2)/\Gamma(x-k+y)$, the following table contains
  the predicted and actual heights of telescopers for various orders and degrees.
  An entry $u|v$ in cell $(r,d)$ means that Thm.~\ref{THM:hyperodh} predicts the existence of a telescoper
  of order~$r$, degree~$d$ and height~$h\geq u$, and that there actually exists a telescoper of order~$r$,
  degree~$d$ and height~$h=v$. A dot indicates that no operator with the corresponding order
  and degree exist or is predicted.
  \[\scriptscriptstyle
  \begin{array}{@{}c@{\kern3pt}|@{\kern3pt}c@{\kern3pt}c@{\kern3pt}c@{\kern3pt}c@{\kern3pt}c@{\kern3pt}c@{\kern3pt}c@{\kern3pt}c@{\kern3pt}c@{\kern3pt}c@{}}
    \llap{$d$ }
    12 & \cdot|\cdot & \cdot|\cdot & 42|9 & 25|6 & 22|6 & 21|5 & \x{22}|5 & 22|5 & 23|4 & 24|4 \\
    11 & \cdot|\cdot & \cdot|\cdot & 50|9 & 27|6 & 24|6 & 23|5 & \x{24}|5 & 24|5 & 25|5 & 26|4 \\
    10 & \cdot|\cdot & \cdot|\cdot & 62|9 & 31|7 & 27|6 & 26|5 & 26|5 & \x{27}|5 & 28|5 & 29|4 \\
    9 & \cdot|\cdot & \cdot|\cdot & 86|9 & 36|7 & 30|6 & 29|5 & 30|5 & \x{30}|5 & 32|5 & 33|5 \\
    8 & \cdot|\cdot & \cdot|\cdot & 158|9 & 45|7 & 36|6 & 35|5 & 35|5 & \x{36}|5 & 37|5 & 39|5 \\
    7 & \cdot|\cdot & \cdot|\cdot & \cdot|9 & 62|7 & 47|6 & 43|6 & 43|5 & \x{44}|5 & 46|5 & 47|5 \\
    6 & \cdot|\cdot & \cdot|\cdot & \cdot|9 & 114|7 & 69|6 & 61|6 & 59|5 & \x{60}|5 & 61|5 & 63|5 \\
    5 & \cdot|\cdot & \cdot|\cdot & \cdot|9 & \cdot|7 & 160|7 & 113|6 & 102|5 & 98|5 & \x{99}|5 & 101|5 \\
    4 & \cdot|\cdot & \cdot|\cdot & \cdot|\cdot & \cdot|7 & \cdot|7 & \cdot|6 & 570|6 & 371|5 & 312|5 & 288|5 \\
    3 & \cdot|\cdot & \cdot|\cdot & \cdot|\cdot & \cdot|10 & \cdot|8 & \cdot|6 & \cdot|6 & \cdot|6 & \cdot|6 & \cdot|6 \\
    2 & \cdot|\cdot & \cdot|\cdot & \cdot|\cdot & \cdot|\cdot & \cdot|\cdot & \cdot|8 & \cdot|8 & \cdot|8 & \cdot|8 & \cdot|8 \\
    1 & \cdot|\cdot & \cdot|\cdot & \cdot|\cdot & \cdot|\cdot & \cdot|\cdot & \cdot|\cdot & \cdot|\cdot & \cdot|\cdot & \cdot|\cdot & \cdot|\cdot \\
    0 & \cdot|\cdot & \cdot|\cdot & \cdot|\cdot & \cdot|\cdot & \cdot|\cdot & \cdot|\cdot & \cdot|\cdot & \cdot|\cdot & \cdot|\cdot & \cdot|\cdot \\\hline
    & 0 & 1 & 2 & 3 & 4 & 5 & 6 & 7 & 8 & 9\rlap{ $r$} \\
  \end{array}
  \]
  Thm.~\ref{THM:hyperodh} inherits from the order-degree curve of~\cite{ChKa2012a} that it correctly predicts
  the orders but overshoots its degrees.
  For example, with $r=2$, the theorem predicts a telescoper of order 2 and degree $d$ 
  whenever $d\geq 8$, and yet we actually found telescopers of order 2 and degree 5, 6 or 7 
  by direct calculation.
  Also the predicted heights overshoot, but note
  that the bound still supports trading order and degree against height, e.g.,
  the height predicted for $(r,d)=(4,6)$ is lower than that for $(r,d)=(4,5)$ and $(r,d)=(3,6)$.
  Note also that
  the predicted heights are not weakly decreasing for fixed $d$ and growing~$r$, as they should
  (leftmost violations are highlighted in bold).
  This is a weakness of Thm.~\ref{THM:hyperodh} caused by the quadratic term in $r$ appearing in the
  inequality. 
\end{example}

\subsection{The rational case}\label{SUBSEC:rat}
In this subsection, we deal with the remaining case where the given
proper hypergeometric term $H$ of the form \eqref{EQ:hyper} can be
further written as $H = fH_0$ for $f\in C(x,y,k)$ and a proper
hypergeometric term $H_0$ with $S_k(H_0)/H_0=1$.  
Let $a,b\in C[x,y,k]$ be such that $S_x(H_0)/H_0 = a/b$. Along similar lines
in the proof of~\cite[Lemma~9]{ChKa2012a}, one can show that whenever 
$f$ admits a telescoper of order $r$, degree $d$ and height $h$, there always
exists a telescoper for $H$ of order $r$, degree at most $d+r\max\{\deg_x(a),\deg_x(b)\}$
and height at most $h+r\max\{\deg_y(a),\deg_y(b)\}$. Based on this transformation,
we may assume without loss of
generality that $H_0 = 1$. In other words, we assume that $H$ is a
proper hypergeometric term of the form \eqref{EQ:hyper} and, at the
same time, is a rational function in $C(x,y,k)$, which is equivalent
to say that $H$ is a rational function whose denominator factors into
integer-linear factors.
Following \cite[\S 4]{ChKa2012a}, we employ the direct algorithm of Le
\cite{Le2003a}, instead of Zeilberger's algorithm, for analyzing sizes of 
telescopers in the rational case so as to obtain a much sharper bound.

Let $H$ be a rational proper hypergeometric term. The algorithm of
Le \cite{Le2003a} first decomposes $H$ as
\begin{equation}\label{EQ:Hdecomp}
H = S_k(g) - g
+ \frac1u\sum_{i=1}^sV_i\Big(\frac1{(a_i x + a_i' k+a_i'')^{e_i}}\Big),
\end{equation}
where $g\in C(x,y,k)$, $u\in C[x,y]$, $V_i\in C[x,y][S_x]$,
$a_i,a_i',e_i\in \set Z$, $a_i''\in C[y]$ with $a_i'>0, e_i>0$,
$\gcd(a_i,a_i') = 1$ for all $i$, and
\[
\Big(\frac{a_i}{a_i'}-\frac{a_j}{a_j'}\Big)x +
\Big(\frac{a_i''}{a_i'}-\frac{a_j''}{a_j'}\Big)
\notin \set Z
\]
for all $i\neq j$ with $e_i=e_j$. Then for $i = 1,\dots,s$, it
computes an operator $L_i\in C(x,y)[S_x]$ such that $S_x^{a_i'}-1$ is
a right divisor of $L_i(\frac1uV_i)$. A common left multiple $L\in
C[x,y][S_x]$ of the operators $L_1,\dots,L_s$ finally leads to a
telescoper for $H$.

Note that the main computational work happens in the last two steps.
It is sensible to assume in the following degree analysis that we
already know the decomposition \eqref{EQ:Hdecomp} and to express the
bounds in terms of quantities from the last term in \eqref{EQ:Hdecomp},
rather than from $H$.

\begin{theorem}\label{THM:ratodh}
Let $H$ be a rational proper hypergeometric term admitting the
decomposition \eqref{EQ:Hdecomp}. For $i = 1,\dots,s$, assume that
$V_i$ has degree $\vartheta_i$ and height $\tau_i$. Let $r,d,h\in \set N$
be such that $d\geq \deg_x(u)$, $h\geq \deg_y(u)$, and 
\begin{align*}
&(r+1)(d+1-\deg_x(u))(h+1-\deg_y(u))-\sum_{i=1}^sa_i'\vartheta_i\tau_i \\
&-(d+1-\deg_x(u))\sum_{i=1}^sa_i'\tau_i
-(h+1-\deg_y(u))\sum_{i=1}^sa_i'\vartheta_i\\
&-(d+1-\deg_x(u))(h+1-\deg_y(u))\sum_{i=1}^sa_i'
>0
\end{align*}
Then there exists a telescoper $L$ for $H$ of order at most~$r$, degree at most~$d$
and height at most~$h$.
\end{theorem}
\begin{proof}
According to the algorithm of Le \cite{Le2003a}, it suffices to find
some $L\in C[x,y][S_x]$ and $R_i\in C(x,y)[S_x]$ with the property
that $L(\frac1uV_i) = R_i(S_x^{a_i'}-1)$ for all $i=1,\dots,s$.

Let $\rho_i$ be the order of $V_i$ for $i = 1,\dots,s$ and make an
ansatz $L = \tilde L u$ with
\[
\tilde L = \sum_{i=0}^r\sum_{j=0}^{d-\deg_x(u)}\sum_{\ell=0}^{h-\deg_y(u)}
c_{i,j,\ell}y^\ell x^jS_x^i.
\]
Then $L$ has order $r$, degree $d$, height $h$, and $L\frac1uV_i
= \tilde L V_i$ for $i = 1,\dots,s$. It thus amounts to constructing
operators $R_i\in C[x,y][S_x]$ such that $\tilde LV_i=R_i(S_x^{a_i'}-1)$.
Note that $\tilde LV_i$ has order $r+\rho_i$, degree
$d-\deg_x(u)+\vartheta_i$ and height $h-\deg_y(u)+\tau_i$. Also note
that $S_x^{a_i'}-1$ has order $a_i'$, degree 0 and height 0. So we
consider ansatzes for the $R_i$ of order $r+\rho_i-a_i'$, degree
$d-\deg_x(u)+\vartheta_i$ and height $h-\deg_y(u)+\tau_i$, respectively.

Comparing coefficients with respect to $x,y$ and $S_x$ in all the
required identities $\tilde LV_i = R_i(S_x^{a_i'}-1)$ leads to a
linear system with
\begin{align*}
&(r+1)(d-\deg_x(u)+1)(h-\deg_y(u)+1)\\
&+\sum_{i=1}^s(r+\rho_i-a_i'+1)(d-\deg_x(u)+\vartheta_i+1)(h-\deg_y(u)+\tau_i+1)
\end{align*}
variables and
\[
\sum_{i=1}^s(r+\rho_i+1)(d-\deg_x(u)+\vartheta_i+1)(h-\deg_y(u)+\tau_i+1)
\]
equations. Since $r,d,h$ satisfy the inequality given in the theorem,
the number of variables exceeds the number of equations, and thus the
resulting linear system has a nontrivial solution.
\end{proof}
\begin{example}\label{EX:rat}
  For the rational function
  \[
  H=\frac1{x+y}(((y+1)^3+1)S_x-(x+y+y^3))(\frac1{x+3k+y}),
  \]
  the heights of telescopers predicted by Thm.~\ref{THM:ratodh} and
  the heights of telescopers actually observed are as follows.
  \[
\scriptscriptstyle
  \begin{array}{@{}c|c@{\kern5pt}c@{\kern5pt}c@{\kern5pt}c@{\kern5pt}c@{\kern5pt}c@{\kern5pt}c@{\kern5pt}c@{\kern5pt}c@{\kern5pt}c@{\kern5pt}c@{\kern5pt}c@{\kern5pt}c@{\kern5pt}c@{\kern5pt}c@{\kern5pt}c@{\kern5pt}c@{}}
      \llap{$d$ }
    7& \cdot|\cdot  & \cdot|\cdot  & \cdot|\cdot  & 19|9 & 7|4 & 5|3 & 3|1 & 3|1 & 2|1 & 2|1 & 2|1 & 2|1 & 2|1 \\
    6& \cdot|\cdot  & \cdot|\cdot  & \cdot|\cdot  & 22|9 & 8|4 & 5|3 & 4|1 & 3|1 & 2|1 & 2|1 & 2|1 & 2|1 & 2|1 \\
    5&\cdot|\cdot  & \cdot|\cdot & \cdot|\cdot  & 28|9 & 8|4 & 5|3 & 4|1 & 3|1 & 3|1 & 2|1 & 2|1 & 2|1 & 2|1 \\
    4& \cdot|\cdot  & \cdot|\cdot  & \cdot|\cdot  & 46|9 & 10|4 & 6|3 & 4|1 & 3|1 & 3|1 & 2|1 & 2|1 & 2|1 & 2|1 \\
    3& \cdot|\cdot  & \cdot|\cdot  & \cdot|\cdot  & \cdot|\cdot  & 13|4 & 7|3 & 5|1 & 4|1 & 3|1 & 3|1 & 2|1 & 2|1 & 2|1 \\
    2& \cdot|\cdot  & \cdot|\cdot  & \cdot|\cdot  & \cdot|\cdot  & 28|4 & 10|3 & 6|1 & 4|1 & 4|1 & 3|1 & 3|1 & 2|1 & 2|1 \\
    1& \cdot|\cdot  & \cdot|\cdot  & \cdot|\cdot  & \cdot|\cdot  & \cdot|\cdot  & \cdot|\cdot  & 19|1  & 10|1 &
      7|1 & 5|1 & 4|1 & 4|1 & 3|1 \\
    0& \cdot|\cdot  & \cdot|\cdot  & \cdot|\cdot  & \cdot|\cdot  & \cdot|\cdot  & \cdot|\cdot  & \cdot|\cdot  &
      \cdot|\cdot  & \cdot|\cdot  & \cdot|\cdot  & \cdot|\cdot  & \cdot|\cdot  & \cdot|\cdot  \\\hline
    & 0 & 1 & 2 & 3 & 4 & 5 & 6 & 7 & 8 & 9 & 10 & 11 & 12 \rlap{ $r$}
  \end{array}
  \]
  Thm.~\ref{THM:ratodh} predicts orders and degrees correctly but the bound on the height overshoots.
  At least in this example, while the prediction supports trading order and degree against height, the actual
  operators seem to only allow for trading order against height. This might be a general phenomenon when $s=1$.
\end{example}

\section{Contraction Ideals}\label{SEC:desing}

\def\con{\operatorname{Con}}%

In~\cite{CJKS2013}, it has been shown that order-degree curves are caused by so-called
removable singularities. This analysis is more general than the results obtained
in the previous sections in so far as it is not limited to operators obtained by
a certain operation, e.g., creative telescoping, but applies to any operator ideal.
At the same time, the ``general'' result below does not include the results stated
above as special cases, because it depends on different assumptions on what is known
about the operators at hand.

In the present section, we make an attempt at extending the analysis of \cite{CJKS2013}
to order-degree-height surfaces. Although there is some theory about the
removability of singularities in Ore algebras of the form $C[x,y][\partial]$~\cite{zhang16,zhang17,CKLZ2019},
we will formulate our result in slightly different terms. For an operator
$L\in C[x,y][\partial]$, let $\<L>$ denote the ideal it generates in
$C(x,y)[\partial]$. Then $\con\<L>:=\<L>\cap C[x,y][\partial]$ is an ideal
of $C[x,y][\partial]$, called the \emph{contraction ideal} of~$\<L>$.
We assume that some elements of the contraction ideal $\con\<L>$ are given
and study how they give rise to order-degree-height surfaces about
the elements of~$\con\<L>$.

The construction proposed in~\cite{CJKS2013} can be summarized as follows.
Suppose we know two operators $L\in C[x][\partial]$ and $L_1\in\con\<L>$.
The goal is an estimate relating the orders $r$ and the degrees $d$ of the elements of~$\con\<L>$.
Suppose that $\deg_\partial(L_1)>\deg_\partial(L)$ and let $p\in C[x]$ and $P\in C[x][\partial]$ be such that $pL_1=PL$.
Suppose further that $\deg_x(p)>\deg_x(\lc_\partial(P))$.
In order to search for elements of $\con\<L>$ of order~$r$ and degree~$d$, make an ansatz
$(Q_0+Q_1\frac1pP)L$ with some operators $Q_0,Q_1\in C[x][\partial]$ with
\begin{alignat*}1
  \deg_\partial(Q_0) &\leq r - \deg_\partial(L),\\
  \deg_x(Q_0) &\leq \deg_x(P) - 1,\\
  \deg_\partial(Q_1) &\leq r - \deg_\partial(L_1),\\
  \deg_x(Q_1) &\leq \deg_x(p) - \deg_x(\lc_\partial(P)) - 1.
\end{alignat*}
Then $\deg_x(Q_0+Q_1\frac1pP)\leq\deg_x(P)-1$, so if we equate the coefficients of
all terms of degree $>d-\deg_x(L)$ to zero, we obtain
\[
  (r-\deg_\partial(L)+1)(\deg_x(P)-1-d+\deg_x(L))
\]
linear constraints on the
\begin{alignat*}1
  &(r-\deg_\partial(L)+1)\deg_x(P)\\
  &+ (r-\deg_\partial(L_1)+1)(\deg_x(p) - \deg_x(\lc_\partial(P)))
\end{alignat*}
undetermined coefficients of $Q_0$ and~$Q_1$. The linear system is underdetermined if
$r\geq\deg_\partial(L)$ and 
\[
d\geq\deg_x(L) -
\Bigl(1 - \frac{\deg_\partial(L_1)-\deg_\partial(L)}{r+1-\deg_\partial(L)}\Bigr)
\bigl(\deg_x(p) - \deg_x(\lc_\partial(P))\bigr).
\]
This matches the formula in Thm.~9 of~\cite{CJKS2013} for the case $m=1$.
Observe that every nonzero solution of the linear system gives rise to a nonzero
operator $Q_0+Q_1\frac1pP$, because the degree restrictions imposed in the ansatz
imply that the leading coefficient of $Q_1\frac1pP$ is a proper rational function while 
$Q_0$ has only polynomial coefficients, and then
$Q_0+Q_1\frac1pP=0\Rightarrow Q_0=Q_1=0$.

For order-degree curves, the argument just sketched generalizes to the case where for
the given $L\in C[x][\partial]$ we know several operators $L_1,\dots,L_m\in\con\<L>$.
However, the general proof given in~\cite{CJKS2013} makes use of the fact that $C[x]$ is a
Euclidean domain. If we turn to operators in $C[x,y][\partial]$ we are faced with the
problem that $C[x,y]$ is not a Euclidean domain. This makes it less clear how the
ansatz must be shaped in order to ensure that a nonzero solution of the linear system
translates into a nonzero ideal element. Indeed, it seems impossible to formulate a
general ansatz whose shape only depends on the orders, degrees, and heights of
the operators involved. Instead, we must also take into account how many syzygies
there are between the leading coefficients of the operators. This quantity, for which
we introduce the following notation, will appear in our order-degree-height surface
formula.

\begin{defi}\label{def:syz}
  For polynomials $u_1,\dots,u_m\in C[x,y]$, let
  \begin{alignat*}1
    &\syz(u_1,\dots,u_m)\\
    &:=\{(q_1,\dots,q_m)\in C[x,y]^m | q_1u_1 + \cdots + q_mu_m=0\}
  \end{alignat*}
  denote their syzygy module. 
  Let $p_1,\dots,p_m\in C[x,y]\setminus\{0\}$ and $P_1,\dots,P_m\in C[x,y][\partial]\setminus\{0\}$.
  Suppose that $\deg_\partial(P_1)\leq\cdots\leq\deg_\partial(P_m)$.
  For $n\in\set N$, let $m_n\in\{1,\dots,m\}$ be maximal such that $\deg_\partial(P_{m_n})\leq n$,
  let $u_\ell:=\sigma^{n-\deg_\partial(P_\ell)}\Bigl(\frac{\lc_\partial(P_\ell)}{p_\ell}\Bigr)$
  for $\ell=1,\dots,m_n$ and define 
  \begin{alignat*}1
    V_n = \Bigl\{ &(q_1,\dots,q_{m_n})\in\syz(u_1,\dots,u_{m_n}) : \\
     & \deg_x q_\ell \leq\deg_x p_\ell - \deg_x\lc_\partial(P_\ell)\text{ and }\\
     & \deg_y q_\ell \leq\deg_y p_\ell - \deg_y\lc_\partial(P_\ell)\text{ for all $\ell$} \Bigr\}.
  \end{alignat*}
  We define $c_n:=\dim_C V_n$.
\end{defi}

By the following lemma, the sequence $c_0,c_1,\dots$ stabilizes at $n=\max_{\ell=1}^m\deg_\partial(P_\ell)$.
In particular, for every operator $L$ and every choice $p_1,\dots,p_m$ and $P_1,\dots,P_m$ as in Def.~\ref{def:syz},
there is a constant $c$ such that $\sum_{n=0}^{r - \deg_\partial(L)} c_n\leq (r-\deg_\partial(L)+1) c$.

\begin{lemma}
  Let $p_1,\dots,p_m$ and $P_1,\dots,P_m$ be as in Def.~\ref{def:syz}.
  If $n\in\set N$ is such that $m_n=m_{n+1}$, then $c_n=c_{n+1}$.
\end{lemma}
\begin{proof}
  Since $\sigma$ is an automorphism, we have
  \begin{alignat*}1
    &(q_1,\dots,q_{m_n})\in\syz(u_1,\dots,u_{m_n})\\
    &\iff(\sigma(q_1),\dots,\sigma(q_{m_n}))\in\syz(\sigma(u_1),\dots,\sigma(u_{m_n}))
  \end{alignat*}
  for all $q_1,\dots,q_{m_n}$. Since $\sigma$ is also assumed to preserve degrees, 
  \[
  \deg q_\ell\leq\deg p_\ell-\deg\lc_\partial(P_\ell)
  \iff 
  \deg\sigma(q_\ell)\leq\deg p_\ell-\deg\lc_\partial(P_\ell)
  \]
  for all $\ell$ and every choice of~$q_\ell$, both for $\deg=\deg_x$ and $\deg=\deg_y$.
  Together with the assumption $m_n=m_{n+1}$, it follows that $\dim_C V_n=\dim_C V_{n+1}$. 
\end{proof}

\begin{thm}\label{thm:contraction}
  Let $L\in C[x,y][\partial]$ and let $L_1,\dots,L_m\in\con\<L>\subseteq C[x,y][\partial]$.
  For all $\ell=1,\dots,m$, let $p_\ell\in C[x,y]\setminus\{0\}$ and $P_\ell\in C[x,y][\partial]$
  be such that $p_\ell L_\ell=P_\ell L$, the $p_\ell$ are pairwise coprime, 
  $\deg_x p_\ell >\deg_x\lc_\partial(P_\ell)$ and $\deg_y p_\ell >\deg_y\lc_\partial(P_\ell)$. 
  Define $\lambda_{x,\ell}$ and $\lambda_{y,\ell}$ by
  \begin{alignat*}1
    \lambda_{x,\ell} &:= \deg_x p_\ell - \deg_x \lc_\partial(P_\ell),\\
    \lambda_{y,\ell} &:= \deg_y p_\ell - \deg_y \lc_\partial(P_\ell)
  \end{alignat*}
  for all $\ell$, and let
  \begin{alignat*}3
    \eta_x&=\max_{\ell=1}^m(\deg_x P_\ell - \deg_x\lc_\partial(P_\ell)),&\quad\mu_x&=\sum_{\ell=1}^m\deg_x(p_\ell),\\
    \eta_y&=\max_{\ell=1}^m(\deg_y P_\ell - \deg_y\lc_\partial(P_\ell)),&\quad\mu_y&=\sum_{\ell=1}^m\deg_y(p_\ell),\\
    \xi_x&=\sum_{\ell=1}^m(\deg_\partial(L_\ell){-}1)\deg_x(p_\ell),&&\null\kern-3em
    \xi_y =\sum_{\ell=1}^m(\deg_\partial(L_\ell){-}1)\deg_y(p_\ell).
  \end{alignat*}
  Let the numbers $c_0,c_1,\dots$ be as in Def.~\ref{def:syz}.

  Let $r,d,h\in\set N$ be such that $r\geq\deg_\partial(L),\deg_\partial(L_1),\dots,\deg_\partial(L_m)$ and
  \begin{alignat*}1
    &(r-\deg_\partial(L)+1)\Bigl(\tilde\eta_x\tilde\eta_y+\sum_{\ell=1}^m \lambda_{x,\ell}\lambda_{y,\ell}\\
    &\quad - (\tilde\eta_x+r\mu_x-\xi_x)(\tilde\eta_y+r\mu_y-\xi_y)\\
    &\quad + (d - \deg_x(L) + 1+r\mu_x-\xi_x)(h - \deg_y(L) + 1+r\mu_y-\xi_y)\Bigr)\\
    &> \sum_{\ell=1}^m \deg_\partial(P_\ell)\lambda_{x,\ell}\lambda_{y,\ell}+ \max_{n=0}^{r - \deg_\partial(L)} c_n,
  \end{alignat*}
  where $\tilde\eta_x=\max(\eta_x,d-\deg_x(L)+1)$ and $\tilde\eta_y=\max(\eta_y,h-\deg_y(L)+1)$.
  Then there exists a $Q\in C(x,y)[\partial]$ such that $QL\in C[x,y][\partial]$ and
  $QL\neq0$ and $\deg_\partial(QL)\leq r$ and $\deg_x(QL)\leq d$ and $\deg_y(QL)\leq h$.
\end{thm}
\begin{proof}
  For $r,d,h\in\set N$ as in the theorem, consider an ansatz
  \[
    Q:=Q_0 + Q_1\frac1{p_1}P_1 + \cdots + Q_m\frac1{p_m}P_m
  \]
  for a left-multiplier $Q$ for~$L$.
  Note that for every choice $Q_0,\dots,Q_m\in C[x,y][\partial]$ we have $QL\in C[x,y][\partial]$.
  Setting $p_0=1$ and $P_0=1$, we can write $Q=\sum_{\ell=0}^m Q_\ell\frac1{p_\ell}P_\ell$.
  For $\ell\geq1$, we choose the degree bounds $\deg_x(Q_\ell)<\deg_x(p_\ell) - \deg_x \lc_\partial(P_\ell)=\lambda_{x,\ell}$
  and $\deg_y(Q_\ell)<\deg_y(p_\ell) - \deg_y \lc_\partial(P_\ell)=\lambda_{y,\ell}$,
  and for $\ell=0$, we choose the degree bounds $\deg_x(Q_0)<\tilde\eta_x$ and $\deg_y(Q_0)<\tilde\eta_y$.
  We also impose the bounds $\deg_\partial(Q_\ell)\leq r - \deg_\partial(P_\ell) - \deg_\partial(L)$ for all~$\ell$.
  We then have $\deg_\partial(QL)\leq r$ as well as $\deg_x(Q)<\tilde\eta_x$ and $\deg_y(Q)<\tilde\eta_y$.
  In order to ensure the desired shape of~$QL$, we equate undesired coefficients of $Q$ to zero.
  Because $\deg_x(QL)=\deg_x(Q)+\deg_x(L)$ and $\deg_y(QL)=\deg_y(Q)+\deg_y(L)$, we want $Q$
  to be such that $\deg_x(Q)\leq d-\deg_x(L)$ and $\deg_y(Q)\leq h-\deg_y(L)$.
  Undesired coefficients are those which violate these degree conditions. 
  Taking into account that the coefficients of $Q$ are rational functions whose common denominator
  divides the polynomial $\prod_{\ell=1}^m\prod_{i=0}^{\deg_\partial Q_\ell}\sigma^i(p_\ell)$, which has
  $x$-degree at most $r\mu_x - \xi_x$ and
  $y$-degree at most $r\mu_y - \xi_y$, we get
  \begin{alignat*}1
    &\underbrace{(r - \deg_\partial(L) + 1)}_{\vbox{\footnotesize\hbox{number of $\partial$-mono-\vphantom{$y$}}\hbox{mials in $Q$}}}
    \Bigl(\underbrace{(\tilde\eta_x+r\mu_x-\xi_x)(\tilde\eta_y+r\mu_y-\xi_y)}_{\vbox{\footnotesize\hbox{number of $x$-$y$-monomials}\hbox{in the numerator of $Q$}}}\\
    &\hphantom{\times\bigl(}{}-\underbrace{(d - \deg_x(L) + 1+r\mu_x-\xi_x)(h - \deg_y(L) + 1+r\mu_y-\xi_y)}_{\text{number of $x$-$y$-monomials
        that are allowed to stay}}\Bigr)
  \end{alignat*}
  equations. The number of unknowns in the ansatz is
  \begin{alignat*}1
    &\sum_{\ell=0}^m (\deg_\partial(Q_\ell)+1)(\deg_x(Q_\ell)+1)(\deg_y(Q_\ell)+1)\\
    &=(r{-}\deg_\partial(L){+}1)\tilde\eta_x\tilde\eta_y + \sum_{\ell=1}^m(r{-}\deg_\partial(P_\ell){-}\deg_\partial(L){+}1)\lambda_{x,\ell}\lambda_{y,\ell}\\
    &=(r{-}\deg_\partial(L){+}1)\Bigl(\tilde\eta_x\tilde\eta_y+\sum_{\ell=1}^m \lambda_{x,\ell}\lambda_{y,\ell}\Bigr)
      -\sum_{\ell=1}^m \deg_\partial(P_\ell)\lambda_{x,\ell}\lambda_{y,\ell}.
  \end{alignat*}
  If the number of variables exceeds the number of equations, the linear system will have nonzero solutions.
  However, a nonzero solution $(Q_0,\dots,Q_m)$ need not translate into a nonzero multiplier
  $Q=\sum_{\ell=0}^m Q_\ell\frac1{p_\ell}P_\ell$.
  There is a danger of cancellations. By the restrictions on the degrees, no cancellation can happen between
  $Q_0$ and $\sum_{\ell=1}^m Q_\ell\frac1{p_\ell}P_\ell$, because $Q_0$ has only polynomial coefficients while
  $\sum_{\ell=1}^m Q_\ell\frac1{p_\ell}P_\ell$ is either zero or the leading coefficient is a proper rational
  function. It remains to avoid that $\sum_{\ell=1}^m Q_\ell\frac1{p_\ell}P_\ell$ is identically zero.
  If it is identically zero, there must in particular be a cancellation among the leading terms.
  Let $n$ be maximal such that at least one of $[\partial^n] Q_\ell\frac1{p_\ell}P_\ell$ is nonzero,
  where we write $[\partial^i]X$ for the coefficient of $\partial^i$ in $X\in C(x,y)[\partial]$.
  Then
  \[
    ([\partial^{n-\deg_\partial(P_1)}]Q_1, \dots, [\partial^{n-\deg_\partial(P_m)}]Q_m)
  \]
  belongs to the vector space~$V_n$ of Def.~\ref{def:syz} (coordinates $\ell$ with $n<\deg_\partial(P_\ell)$
  are meant to be omitted).
  For each $n$ there can be at most $c_n$ many linearly independent solutions for which a cancellation happens.
  Therefore, if the number of variables exceeds the number of equations by more than $\max_{n=0}^{r - \deg_\partial(L)} c_n$,
  there must be at least one solution that does not completely cancel.
  Since this is the case if $r,d,h$ are chosen as specified in the theorem, we are done. 
\end{proof}

\begin{example}
  Let $L\in C[x,y][S_x]$ be the minimal telescoper of the hypergeometric term $H=k\Gamma(x+k+y^2)/\Gamma(x-k+y)$
  already considered in Example~\ref{EX:nonrat}.
  It has order~2, degree~5, and height~9, and we have
  \[
   \lc_\partial(L) = (2 x+y^2+y)(x^2+x y^2+x y+y^3-1).
  \]
  There is an $L_1$ of $\con\<L>$ of order~3, degree~8, height~8 and with
  \[
    \lc_\partial(L_1) = 6 x^2+6 x y^2+6 x y+6 x+y^4+4 y^3+4 y^2+3 y.
  \]
  Applying Thm.~\ref{thm:contraction} to $L$ and $L_1$ gives the following height predictions
  for elements of $\con\<L>$ of various orders $r$ and degrees~$d$ (in comparison to the actual
  smallest heights for elements of the respective shape):
  \[
\scriptscriptstyle
  \begin{array}{@{}c|c@{\kern5pt}c@{\kern5pt}c@{\kern5pt}c@{\kern5pt}c@{\kern5pt}c@{\kern5pt}c@{\kern5pt}c@{\kern5pt}c@{\kern5pt}c@{\kern5pt}c@{\kern5pt}c@{}}
    \llap{$d$ }10&\cdot|\cdot & \cdot|\cdot & \cdot|9 & 8|7 & 8|6 & 8|5 & 8|5 & 8|5 \\
    9&\cdot|\cdot & \cdot|\cdot & \cdot|9 & 8|7 & 8|6 & 8|5 & 8|5 & 8|5 \\
    8&\cdot|\cdot & \cdot|\cdot & \cdot|9 & 8|7 & 8|6 & 8|5 & 8|5 & 8|5 \\
    7&\cdot|\cdot & \cdot|\cdot & \cdot|9 & 10|7 & \x{12}|6 & 13|6 & 15|5 & 17|5 \\
    6&\cdot|\cdot & \cdot|\cdot & \cdot|9 & 13|7 & \x{18}|6 & 23|6 & 28|5 & 33|5 \\
    5&\cdot|\cdot & \cdot|\cdot & \cdot|9 & 23|7 & \x{38}|7 & 53|6 & 68|5 & 83|5 \\
    4&\cdot|\cdot & \cdot|\cdot & \cdot|\cdot & \cdot|7 & \cdot|7 & \cdot|6 & \cdot|6 & \cdot|5 \\
    3&\cdot|\cdot & \cdot|\cdot & \cdot|\cdot & \cdot|10 & \cdot|8 & \cdot|6 & \cdot|6 & \cdot|6 \\
    2&\cdot|\cdot & \cdot|\cdot & \cdot|\cdot & \cdot|\cdot & \cdot|\cdot & \cdot|8 & \cdot|8 & \cdot|8 \\
    1&\cdot|\cdot & \cdot|\cdot & \cdot|\cdot & \cdot|\cdot & \cdot|\cdot & \cdot|\cdot & \cdot|\cdot & \cdot|\cdot \\
    0&\cdot|\cdot & \cdot|\cdot & \cdot|\cdot & \cdot|\cdot & \cdot|\cdot & \cdot|\cdot & \cdot|\cdot & \cdot|\cdot \\\hline
    & 0 & 1 & 2 & 3 & 4 & 5 & 6 & 7\rlap{ $r$}
  \end{array}
  \]
  We see that Thm.~\ref{thm:contraction} overshoots but rightly indicates that order and degree
  can be traded against height, for example for $(r,d)=(3,8)$, the predicted height is~$8$, which
  is less than the predicted height for $(r,d)=(3,7)$. 
  Also, similar to Thm.~\ref{THM:hyperodh}, for fixed $d$ and increasing $r$ the minimal predicted heights
  are not weakly decreasing.
  Like in Thm.~\ref{THM:hyperodh}, we blame the quadratic term in $r$ appearing in the inequality for this
  behaviour.
  Thm.~\ref{thm:contraction} leads to better estimates than Thm.~\ref{THM:hyperodh} because knowing $L$ and $L_1$
  in advance is a much stronger assumption than knowing the hypergeometric term~$H$.
\end{example}

Thm.~\ref{thm:contraction} appears to be a natural generalization of the order-degree curve from~\cite{CJKS2013}.
While this order-degree curve turns out to make quite accurate degree predictions, Thm.~\ref{thm:contraction}
is not that sharp. As mentioned in the beginning, an order-degree curve for an operator~$L$
emerges if there are $L_1,P,p$ with $pL_1=PL$ and $\deg_x(p)>\deg_x(\lc_\partial(P))$. In the univariate
case, we can arrange this situation whenever $\lc_\partial(L)$ has a removable factor. It can
in fact always be arranged that $\lc_\partial(P)=1$. This is no longer true for operators in
$C[x,y][\partial]$, because $C[x,y]$ is not a Euclidean domain. As a consequence, it can happen
(and seems to be common) that a factor of $\lc_\partial(L)$ can only be removed at the cost of
introducing another factor into the leading coefficient.

Another interesting difference between Thm.~\ref{thm:contraction} and Thm.~9 of~\cite{CJKS2013} is
that $\deg_x(P)$ cancels out in the derivation of the order-degree curve (cf. the summary at the
beginning of this section), but it no longer cancels in the more general setting of Thm.~\ref{thm:contraction}.
This difference is also responsible for the disturbing quadratic dependence on~$r$, which has no
counterpart in the univariate case. 

\section{Conclusion}

We have shown several situations that order and degree can be traded against height.
The effect was observed experimentally with actual operators and is supported by the theoretical
bounds we have given.
We believe that our results will be useful for deriving refined complexity analyses for
operations involving operators that take heights into account.

Our results are limited to the case of operators in $C[x,y][\partial]$ with $y$-degree
as height, and it would be interesting to have analogous results for operators in $\set Z[x][\partial]$
with integer bitlength as height. We conclude the paper with an example indicating that
trading effects can also be expected in this setting.

\begin{example}
  For $a_n=\sum_{k=0}^n(\binom{n}{2k}+\binom{2n}{k}^2)$ let $a(x)=\sum_{n=0}^\infty a_nx^n$.
  Using LLL, we searched for differential operators of various orders~$r$ and degrees~$d$ that
  annihilate the series $a(x)$ and involve short integers. The results are shown in the table
  below; a number $h$ in cell $(r,d)$ of the table indicates that we found an operator of
  size $(r,d)$ which only involves integers with at most $h$ decimal digits. 
  
  Note that since multiplying an operator from the left by a power of $x$ increases the
  degree of the operator without changing any of its coefficients, the minimal heights
  for a fixed $r$ must be weakly decreasing while moving along the positive direction of $d$-axis. The height~9 we observe
  for order~12 and degree~14 contradicts this expectation. We believe that this is due
  to the fact that LLL is not guaranteed to find the shortest vector of a given lattice.
  For a fixed $d$ and increasing~$r$, it is conceivable that the minimal height increases.
  The increase we observe at order~14 and degree~7 however also seems to be an artifact
  of the LLL-computation, as the next few heights of degree~7 are again~9.
  Same remarks apply to some other cells, e.g., $(12,9)$, $(14,15)$, etc.
  \[\scriptscriptstyle
     \begin{array}{@{}c@{\kern3pt}|@{\kern3pt}c@{\kern3pt}c@{\kern3pt}c@{\kern3pt}c@{\kern3pt}c@{\kern3pt}c@{\kern3pt}c@{\kern3pt}c@{\kern3pt}c@{\kern3pt}c@{\kern3pt}c@{\kern3pt}c@{\kern3pt}c@{\kern3pt}c@{\kern3pt}c@{\kern3pt}c@{\kern3pt}c@{}}
    \llap{$d$ }20 & \cdot & 10 & 9 & 8 & 8 & 8 & 8 & 8 & 9 & 9 & 9 & 10 & 11 & 11 & 11 & 11 & 12 \\
    19 & \cdot & 10 & 9 & 8 & 8 & 8 & 8 & 8 & 9 & 9 & 9 & 10 & 10 & 11 & 11 & 11 & 12 \\
    18 & \cdot & 10 & 9 & 8 & 8 & 8 & 8 & 9 & 9 & 9 & 9 & 10 & 10 & 11 & 11 & 11 & 11 \\
    17 & \cdot & 10 & 9 & 8 & 8 & 8 & 8 & 8 & 9 & 9 & 9 & 9 & 10 & 10 & 10 & 11 & 11 \\
    16 & \cdot & 10 & 9 & 8 & 8 & 8 & 8 & 8 & 8 & 9 & 9 & 10 & 9 & 10 & 10 & 10 & 11 \\
    15 & \cdot & 10 & 9 & 8 & 8 & 8 & 8 & 8 & 9 & 8 & 9 & 9 & 10 & 10 & 10 & 10 & 11 \\
    14 & \cdot & 10 & 9 & 8 & 8 & 8 & 8 & 8 & 9 & 8 & 8 & 9 & 8 & 10 & 10 & 10 & 11 \\
    13 & \cdot & 10 & 9 & 8 & 8 & 8 & 8 & 8 & 8 & 8 & 8 & 8 & 8 & 8 & 10 & 10 & 10 \\
    12 & \cdot & 10 & 10 & 8 & 8 & 8 & 8 & 8 & 8 & 8 & 8 & 8 & 8 & 8 & 8 & 10 & 8 \\
    11 & \cdot & \cdot & 11 & 9 & 8 & 8 & 8 & 8 & 8 & 8 & 8 & 8 & 8 & 8 & 8 & 8 & 8 \\
    10 & \cdot & \cdot & 12 & 9 & 8 & 8 & 8 & 8 & 8 & 8 & 8 & 8 & 8 & 8 & 8 & 8 & 8 \\
    9 & \cdot & \cdot & 14 & 10 & 9 & 8 & 8 & 8 & 9 & 8 & 8 & 8 & 8 & 8 & 8 & 8 & 8 \\
    8 & \cdot & \cdot & 27 & 11 & 10 & 9 & 9 & 9 & 9 & 9 & 9 & 9 & 9 & 9 & 9 & 9 & 9 \\
    7 & \cdot & \cdot & \cdot & 17 & 12 & 10 & 10 & 9 & 9 & 9 & 10 & 9 & 9 & 9 & 9 & 9 & 9 \\
    6 & \cdot & \cdot & \cdot & \cdot & 27 & 15 & 12 & 11 & 11 & 11 & 11 & 10 & 11 & 11 & 10 & 11 & 10 \\
    5 & \cdot & \cdot & \cdot & \cdot & \cdot & \cdot & \cdot & \cdot & 11 & 11 & 11 & 11 & 11 & 11 & 11 & 11 & 11 \\
    4 & \cdot & \cdot & \cdot & \cdot & \cdot & \cdot & \cdot & \cdot & \cdot & \cdot & \cdot & \cdot & \cdot & \cdot & \cdot & \cdot & \cdot \\ \hline
      & 4 & 5 & 6 & 7 & 8 & 9 & 10 & 11 & 12 & 13 & 14 & 15 & 16 & 17 & 18 & 19 & 20\rlap{ $r$} \\
     \end{array}
  \]
\end{example}

\noindent\textbf{Acknowledgement.}
We thank the referees for their careful reading and their valuable suggestions. 

\bibliographystyle{plain}
\bibliography{scrfs}

\end{document}